\documentclass[runningheads]{llncs}

\usepackage{xspace}
\usepackage[utf8]{inputenc}
\usepackage{graphicx}
\usepackage{amsmath}
\usepackage{float}
\usepackage{amsfonts}
\usepackage{algorithm}      
\usepackage{algpseudocode}  
\usepackage{enumerate}
\usepackage{paralist}
\usepackage{enumitem}
\usepackage{hyperref}

\def\rankop{\textsf{rank}}
\def\selop{\textsf{select}}
\def\accessop{\textsf{access}}

\newcommand{\forbMMOne}{3\text{-}14\text{-}2}
\newcommand{\forbMMTwo}{2\text{-}41\text{-}3}

\newcommand{\rminmax}{\textsc{RMinMax}}
\newcommand{\rmin}{\textsc{RMin}}
\newcommand{\rmax}{\textsc{RMax}}
\newcommand{\rtopk}{\textsc{RTopK}}

\newcommand{\opset}{\mathcal{C}}
\newcommand{\Oh}{\mathcal{O}}

\newcommand{\polylog}{\text{polylog}}

\title{Optimal Encodings for Range Top-\texorpdfstring{$k$}{k},
  Selection, and Min-Max}

\author{Pawe{\l} Gawrychowski\inst{1}\thanks{Currently holding a post-doctoral position at Warsaw Center of Mathematics and Computer Science.}  \and Patrick K. Nicholson\inst{2}}
\institute{Institute of Informatics, University of Warsaw, Poland \and Max-Planck-Institut für Informatik, Saarbrücken, Germany}

\allowdisplaybreaks

\begin{document}

\pagestyle{plain}
\maketitle

\begin{abstract}
We consider encoding problems for range queries on arrays. In these
problems the goal is to store a structure capable of recovering the
answer to all queries that occupies the information theoretic minimum
space possible, to within lower order terms.  As input, we are given
an array $A[1..n]$, and a fixed parameter $k \in [1,n]$.  A
\emph{range top-$k$} query on an arbitrary range $[i,j] \subseteq
     [1,n]$ asks us to return the ordered set of indices $\{\ell_1,
     ..., \ell_{k}\}$ such that $A[\ell_m]$ is the $m$-th largest
     element in $A[i..j]$, for $1 \le m \le k$.  A \emph{range
       selection} query for an arbitrary range $[i,j] \subseteq [1,n]$
     and query parameter $k' \in [1,k]$ asks us to return the index of
     the $k'$-th largest element in $A[i..j]$.  We completely resolve
     the space complexity of both of these heavily studied
     problems---to within lower order terms---for all $k = o(n)$.
     Previously, the constant factor in the space complexity was known
     only for $k=1$.  We also resolve the space complexity of another
     problem, that we call \emph{range min-max}, in which the goal is
     to return the indices of both the minimum and maximum elements in
     a range.
\end{abstract}

\section{Introduction}

Many important algorithms make use of range queries over arrays of
values as subroutines~\cite{N13,S13}.  As a prime example, text
indexes that support pattern matching queries often maintain an array
storing the lengths of the longest common prefixes between consecutive
suffixes of the text.  During a search for a pattern this array is
queried in order to find the position of the minimum value in a given
range.  That is, a subroutine is needed that can preprocess an array
$A$ in order to answer \emph{range minimum queries}.  Formally, as
input to such a query we are given a range $[i,j] \subseteq [1,n]$,
and wish to return the index $k = \arg\min_{i \le \ell\le j}A[\ell]$.
In text indexing applications memory is often the constraining factor,
so the question of how many bits are needed to answer range minimum
queries has been heavily studied.  After a long line of research
(see~\cite{BFPSS05,S07}), it has been determined that such queries can
be answered in constant time, by storing a data structure of size $2n
+o(n)$ bits~\cite{FH11}.  Furthermore, this space bound is optimal to
within lower order terms (see~\cite[Sec.~1.1.2]{FH11}).  The
interesting thing is that the space does not depend on the number of
bits required to store individual elements of the array $A$.  After
constructing the data structure we can discard the array $A$, while
still retaining the ability to answer range minimum queries.

Results of this kind, where it is shown that the solutions to all
queries can be stored using less space than is required to store the
original array, fall into the category of \emph{encodings}, and, more
generally, \emph{succinct} data structures~\cite{J89}.  Specifically,
given a set of combinatorial objects $\chi$ we wish to represent an
arbitrary member of $\chi$ using $\lg|\chi| + o(\lg|\chi|)$
bits\footnote{We use $\lg x$ to denote $\log_2 x$.}, while still
supporting queries, if possible.  If queries can be supported by the
representation then we refer to it as a data structure, but if not,
then we refer to it as an encoding.  For the case of range minimum
queries or range maximum queries, the set $\chi$ turns out to be
\emph{Cartesian trees}, which were introduced by Vuillemin~\cite{V80}.
For a given array $A$, the Cartesian tree encodes the solution to all
range minimum queries, and similarly, if two arrays have the same
solutions to all range minimum queries, then their Cartesian trees are
identical~\cite{FH11}.

Recently, there has been a lot of interest the following two problems,
that generalize range maximum queries in two different ways.  The
input to each of the following problems is an array $A[1..n]$, that we
wish to preprocess into an encoding occupying as few bits as possible,
such that the answers to all queries are still recoverable.  We assume
a value $k \ge 1$ is fixed at preprocessing time.
\begin{itemize}

\item \textbf{Range top-$k$:} Given an arbitrary query range $[i,j]
  \subseteq [1,n]$ and $k' \in [1,k]$, return the indices of the $k'$
  largest values in $[i,j]$.  This problem is the natural
  generalization of range maximum queries and has been the focus of a
  several papers, leading to asymptotically optimal lower and upper
  space bounds of $\Omega(n \lg k)$ and $\Oh(n \lg k)$ bits, proved by
  Grossi et al.~\cite{GINRS13} and Navarro, Raman, and
  Rao~\cite{NRS14}, respectively.  The latter upper bound is a data
  structure that can answer range top-$k'$ queries in optimal
  $\Oh(k')$ time.

\item \textbf{Range $k$-selection:} Given an arbitrary query range
  $[i,j] \subseteq [1,n]$ and $k' \leq k$, return the index of the
  $k'$-th largest value in $[i,j]$.  This problem was studied in a
  series of recent papers (see~\cite{GPT09} and \cite{BGJS10} for
  further references), culminating in data structures that occupy a
  linear number of words, and can answer queries in $\Oh(\lg k' / \lg
  \lg n + 1)$ time~\cite{CW13}.  This query time matches a cell-probe
  lower bound for near-linear space data structures~\cite{JL11}. It is
  straightforward to see that any encoding of range top-$k$ queries is
  also an encoding for range $k$-selection queries, though the
  question of how much time is required during a query remains
  unclear~\cite{NRS14}.  Very recently, Navarro, Raman, and
  Rao~\cite{NRS14} described a data structure that can be used to
  answer range $k$-selection queries in optimal $\Oh(\lg k' / \lg \lg n
  + 1)$ time~\cite{NRS14}, and, like the range top-$k$ data structure,
  occupies $\Oh(n \lg k)$ bits of space.
\end{itemize}

\subsubsection{Our Results}

We present the first space-optimal encodings to range top-$k$---and
therefore range selection also---as well as a new problem that we call
\emph{range min-max}, in which the goal is to return the indices of
both the minimum and maximum element in the array.  We emphasize that,
on their own, the encodings for range top-$k$ and selection do not
support queries efficiently: they merely store the solutions to all
queries in a compressed form.  However, our encoding for range min-max
can be augmented with $o(n)$ additional bits of data to create a data
structure that supports queries in $\Oh(1)$ time.  Furthermore, even
without query support, our encodings for range top-$k$ and selection
address a problem posed in the papers of Grossi et al.~ \cite{GINRS13}
and Navarro et al.~\cite{NRS14}.

In Table~\ref{tab:results} we present a summary of previous and new
results.  Prior to this work, the only value for which the exact
coefficient of $n$ was known was the case in which $k=1$ (i.e., range
maximum queries).  For even $k=2$ the best previous estimate was that
the coefficient of $n$ is between $2.656$ and
$3.272$~\cite{PNRR14}. The lower bound of $2.656$ was derived using
generating functions and an extensive computational
search~\cite{PNRR14}.  In contrast, our method is purely combinatorial
and gives the exact coefficient for all $k = o(n)$.  For $k=2,3,4$ the
coefficients are (rounding up) $2.755$, $3.245$, and $3.610$,
respectively.

\begin{table}[!t]
\centering
\caption{\label{tab:results} Old and new results. Both upper and lower
  bounds are expressed in bits.  Our bounds make use of the binary
  entropy function $H(x) = x \lg(\frac{1}{x}) + (1-x)
  \lg(\frac{1}{1-x})$.  For the entry marked with a $\dagger$ the
  claimed bound holds when $k = o(n)$.}
\begin{tabular}{|l|l|l|l|l|}
\hline
Ref. & Query & Lower Bound & Upper Bound & Query Time \\
\hline
\cite{FH11}   & max     & $2n - \Theta(\lg n)$ & $2n + o(n)$ & $\Oh(1)$ \\ 
\cite{GINRS13,NRS14}     & top-$k$ & $\Omega(n \lg k)$ & $\Oh(n \lg k)$ & $\Oh(k')$ \\ 
\cite{PNRR14}     & top-$2$ & $2.656n - \Theta(\lg n)$ & $3.272n + o(n)$ & $\Oh(1)$ \\
\hline
Thm.~\ref{thm:min-max-ub},~\ref{thm:min-max-lb}     & min-max & $3n - \Theta(\lg(n))$ & $3n + o(n)$ & $\Oh(1)$ \\ 
Thm.~\ref{thm:top-k-ub},~\ref{thm:top-k-lb}     & top-$2$ & $3nH(\frac{1}{3}) - \Theta(\polylog(n))$ & $3nH(\frac{1}{3}) + o(n)$  & --- \\ 
Thm.~\ref{thm:top-k-ub},~\ref{thm:top-k-lb}     & top-$k$ & $(k+1)nH(\frac{1}{k+1})(1-o(1))\dagger$  & $(k+1)nH(\frac{1}{k+1}) + o(n)$ & --- \\
\hline
\end{tabular}
\end{table}

As mentioned above, a negative aspect of our encodings is that they
appear to be somewhat difficult to use as the basis for a data
structure.  However, in Section~\ref{sec:datastructure}, we present a
data structure based on our encoding that \emph{nearly} matches the
optimal space bound.  Explicitly, we can achieve a space bound of
$(k+1.5)nH(\frac{1.5}{k+1.5}) + o(n\lg k)$ bits with query time
$\Oh(\text{poly}(k\lg n))$.  Thus, our data structure achieves space
much closer to the optimal bound than the previous best
result~\cite{NRS14}, but the query time is worse.  We leave the
following data structure problem open: how can range top-$k$ and
selection queries be supported with optimal query time using space
matching our encodings (to within lower order terms)?

Finally, we wish to point out that although our formulation of the
range top-$k$ problem returns the indices in sorted order, the
constant factor in our lower bound also holds for the \emph{unsorted}
version, in which we return the indices in an arbitrary order,
provided $k=o(n)$.  This follows since any encoding strategy for
unsorted range top-$k$ can be used to construct a sorted top-$k$
encoding, by padding the end of the input array with $k-1$ values
larger than any other.  The unsorted encoding of this padded array can
be used to infer the solution to an arbitrary sorted top-$k$ query
$[i,j]$ by examining the solutions to queries $[i,j], [i,j+1], ...,
[i,n+k-1]$: see Appendix~\ref{sec:reduction} for details.


\subsubsection{Discussion of Techniques and Road Map} 

Prior work for top-$k$, for $k \ge 2$, focused on encoding a
decomposition of the array, called a shallow
cutting~\cite{GINRS13,NRS14}.  Since shallow cuttings are a general
technique used to solve many other range searching
problems~\cite{M92,JL11}, these previous works~\cite{GINRS13,NRS14}
required additional information beyond storing the shallow cutting in
order to recover the answers to top-$k$ queries.  Furthermore, in
these works the exact constant factor is not disclosed, though we
estimate it to be at least twice as large as the bounds we
present. For the specific case of range top-$2$ queries a different
encoding has been proposed based on \emph{extended Cartesian
  trees}~\cite{PNRR14}.  In contrast to both of the previous
approaches, our encoding is based the approach of Fischer and
Heun~\cite{FH11}, who describe what is called a 2D min-heap
(resp. max-heap) in order to encode range minimum queries (resp. range
maximum queries).  We begin in Section~\ref{sec:min-max} by showing
how to generalize their technique to simultaneously answer both range
minimum and range maximum queries.  Our encoding provides the answer
to both using $3n +o(n)$ bits in total, compared to $4n +o(n)$ bits
using the trivial approach of constructing both encodings separately.
We then show this bound is optimal by proving that any encoding for
range min-max queries can be used to distinguish a certain class of
permutations.  We move on in Section~\ref{sec:top-k} to generalize
Fischer and Heun's technique in a clean and natural way to larger
values of $k$.  Indeed, the encoding we present---like that of Fischer
and Heun---is simple enough to implement.  The main difficulty is
proving that the bound achieved by our technique is optimal.  For this
we enumerate a particular class of walks, via an application of the
so-called cycle lemma of Dvoretzky and Motzkin~\cite{DM47}.

Finally, in Section~\ref{sec:datastructure} we show that our encoding
can be used as the basis for a range top-$k$ data structure.  Though
the resultant space bound and query time are suboptimal, we note that
interesting challenges had to be overcome to design a data structure
based on our encoding.  Concisely, we required the ability to
decompose the encoding into smaller blocks in order to support queries
efficiently.  To do this we, in some sense, generalized the pioneers
approach of Jacobson~\cite{J89} via a non-trivial decomposition
theorem.  Since balanced parentheses representations appear in many
succinct data structures, we believe this will likely be of
independent interest.

\section{Optimal Encodings of Range Min-Max Queries\label{sec:min-max}}

In this section we describe our encoding for range min-max queries.
We use $\rminmax(A[i..j])$ to denote a range min-max query on a
subarray $A[i..j]$.  The solution to the query is the ordered set of
indices $\{\ell_1, \ell_2\}$ such that $\ell_1 = \arg\max_{\ell \in
  [i,j]}A[\ell]$ and $\ell_2 = \arg\min_{\ell \in [i,j]}A[\ell]$.

\subsection{Review of Fischer and Heun's Technique}

We review the algorithm of Fischer and Heun~\cite{FH11} for
constructing the encoding of range minimum (resp. maximum) queries.

\begin{figure}
\centering
\includegraphics[width=\textwidth]{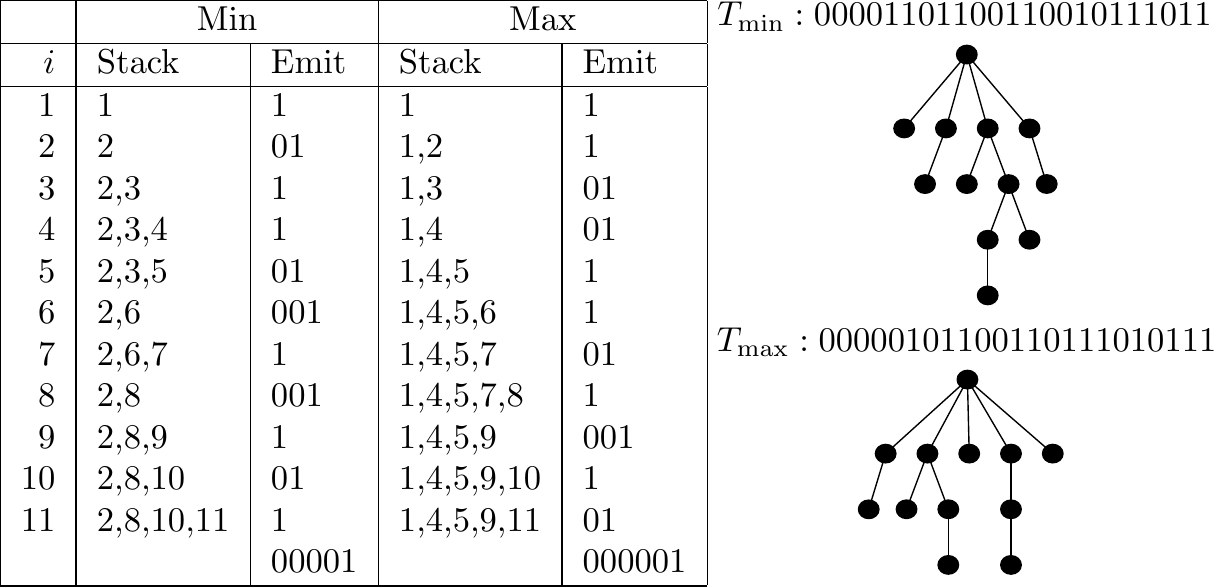}
\caption{\label{fig:perm}A trace of Fischer and Huen's algorithm for
  constructing the encoding for range minimum and maximum queries on
  an array $A[1..11]=(11,1,7,10,9,3,4,2,8,5,6)$.}
\end{figure}

Consider an array $A[1..n]$ storing $n$ numbers. Without loss of
generality we can alter the values of the numbers so that they are a
permutation, breaking ties in favour of the leftmost element. To
construct the encoding for range minimum queries we sweep the array
from left to right\footnote{In the original paper the sweeping process
  moves from right to left, but either direction yields a correct
  algorithm by symmetry.}, while maintaining a stack. A string of bits
$T_{\min}$ (resp. $T_{\max}$) will be emitted in reverse order as we
scan the array.  Whenever we push an element onto the stack, we emit a
one bit, and whenever we pop we emit a zero bit. Initially the stack
is empty, so we push the position of the first element we encounter on
the stack, in this case, $1$.  Each time we increment the current
position, $i$, we compare the value of $A[i]$ to that of the element
in the position $t$, that is stored on the top of the stack.  While
$A[t]$ is not less than (resp. not greater than) $A[i]$, we pop the
stack.  Once $A[t]$ is less than (resp. greater than) the current
element or the stack becomes empty, we push $i$ onto the stack.  When
we reach the end of the array, we pop all the elements on the stack,
emitting a zero bit for each element popped, followed by a one bit.
An example illustrating a trace of the algorithm described here can be
found in Figure~\ref{fig:perm}.

Fischer and Heun showed that the string of bits output by this process
can be used to encode a rooted ordinal tree in terms of its
\emph{depth first unary degree sequence} or DFUDS~\cite{FH11}.  To
extract the tree from a sequence, suppose we read $d$ zero bits until
we hit the first one bit.  Based on this, we create a node $v$ of
degree $d$, and continue building first child of $v$ recursively.
Since there are at most $2n$ stack operations, the tree is therefore
represented using $2n$ bits.  We omit the technical details of how a
query is answered, but the basic idea is to augment this tree
representation with succinct data structures supporting navigation
operations.  The following corollary summarizes part of their result:

\begin{lemma}[Corollary~5.6~\cite{FH11}]\label{lem:aux-index}
Given the DFUDS representation of $T_{\min}$ (resp. $T_{\max}$) any
query $\rmin(A[i..j])$ (resp. $\rmax(A[i..j])$) can be
answered in constant time using an index occupying $\Oh(\frac{n \log
  \log n}{ \log n}) = o(n)$ additional bits of space.
\end{lemma}

\subsection{Upper Bound for Range Min-Max Queries}

We propose the following encoding for a simultaneous representation of
$T_{\min}$ and $T_{\max}$.  Scan the array from left to right and
maintain two stacks: a min-stack for range minimum queries, and a
max-stack for range maximum queries.  Notice that in each step except
for the first and last, we are popping an element from exactly one of
the two stacks.  This crucial observation allows us to save space.  We
describe our encoding in terms of the min-stack and the max-stack
maintained as above.  Unlike before however, we maintain two separate
bit strings, $T$ and $U$. If the new element causes $\delta \ge 1$
elements on the min-stack to be popped, then we prepend
$0^{\delta-1}1$ to the string $T$, and prepend $0$ to the string $U$.
Otherwise, if the new element causes $\delta$ elements on the
max-stack to be popped, we prepend $0^{\delta-1}1$ to the string $T$,
and $1$ to the string $U$.  Since exactly $2n$ elements are popped
during $n$ push operations, the bit string $T$ has length $2n$, and
the bit string $U$ has length $n$, for a total of $3n$ bits.

Before stating our theorem, we require the following result by Raman,
Raman, and Rao~\cite{RRR07}:

\begin{lemma}[\cite{RRR07}]\label{lem:rrr} Let $\mathcal{V}$ be a
  bit vector of length $n$ bits, containing $m$ one bits.  In the
  word-RAM model with word size $\Theta(\lg n)$ bits, there is a data
  structure of size $\lg \binom{n}{m} + \Oh(\frac{n\lg \lg n}{\lg n})
  \le n H(\frac{m}{n}) + \Oh(\frac{n \lg \lg n}{\lg n})$ bits that
  supports the following operations in $\Oh(1)$ time, for any $i \in
  [1,n]$:
\begin{enumerate}
\item $\accessop(\mathcal{V}, i)$: return the bit at index $i$ in $\mathcal{V}$.
\item $\rankop_\alpha(\mathcal{V}, i)$: return the number of bits with
  value $\alpha \in \{0,1\}$ in $\mathcal{V}[1..i]$.
\item $\selop_\alpha(\mathcal{V}, i)$: return the index of the $i$-th
  bit with value $\alpha \in \{0,1\}$.
\end{enumerate}
\end{lemma}

Next, we show that by using our encoding and Lemma~\ref{lem:rrr} it
is possible to also support queries on this encoding in $\Oh(1)$ time.

\begin{theorem}\label{thm:min-max-ub}
There is a data structure that occupies $3n + o(n)$ bits of space,
such that any query $\rminmax(A[i..j])$ can be answered in $\Oh(1)$
time.
\end{theorem}

\begin{proof}
By Corollary~\ref{lem:aux-index}, to prove the theorem, it is
sufficient to show that there is a data structure that occupies $3n +
o(n)$ bits of space, and can recover any block of $\lg n$ consecutive
bits from both $T_{\min}$ and $T_{\max}$ in $\Oh(1)$ time.

If we have such a structure that can extract any
block from either DFUDS representation, then we can use it as an
oracle to access the DFUDS representation of either tree.  Thus, we
need only apply Lemma~\ref{lem:aux-index} to complete the theorem.
The data structure makes use of the bit strings $T$ and $U$, as well
as the following auxiliary data structures:

\begin{enumerate}

\item We precompute a lookup table $\mathcal{L}$ of size
  $\Theta(\sqrt{n} \lg n)$ bits.  The lookup table takes two bit
  strings as input, $s_1$ and $s_2$, both with length $\frac{\lg
    n}{4}$, as well as a single bit $b$.  We conceptually think of the
  bit string $s_1$ as having the format
  $0^{\gamma_1}10^{\gamma_2}1...0^{\gamma_{t-1}}10^{\gamma_{t}}1$,
  where each $\gamma_i \ge 0$. The table returns a new bit string
  $s_3$, of length no greater than $\frac{\lg n}{4}$, that we will
  define next.  Let $\cdot$ be the concatenation operator, and define
  the function:
  $$f(x,y,y') = \begin{cases}0\cdot x & \text{if $y = y'$} \\ 1 &
    \text{otherwise.} \end{cases}$$ If $u_i = 0^{\gamma_i}1$ then $s_3
  = f(u_1,s_2[1],b)\cdot f(u_2,s_2[2],b) \cdots f(u_k,s_2[k],b)$, and
  $s_2[i]$ denotes the $i$-th bit of $s_2$.  Such a table occupies no
  more than the claimed amount of space, and can return $s_3$ (as well
  as $k$) in $\Oh(1)$ time.

\item Each bit in $T$ corresponds to at least one bit in $T_{\min}$ or
  $T_{\max}$.  Also recall that at each step during preprocessing we
  append the value $\delta-1$ in unary to $T$ rather than $\delta$ (as
  in the representation of Fischer and Heun).  Thus, we can treat each
  push operation (with the exception of the first and last)
  corresponding to a single one bit in $T$ as representing three bits:
  two bits in $T_{\min}$ and one bit in $T_{\max}$ or two bits in
  $T_{\max}$ and one bit in $T_{\min}$, depending on the corresponding
  value in $U$.  We store a bit vector $B_{\min}$ of length $2n$ which
  marks the position in $T$ of the bit corresponding to the $(i \lg n
  + 1)$-th bit of $T_{\min}$, for $0 \le i \le \lfloor \frac{2n}{\lg
    n}\rfloor$.  We do the analogous procedure for $T_{\max}$ and call
  the resulting bit vector $B_{\max}$.

\end{enumerate}

Suppose now that we support the operations rank and select on
$B_{\min}$, $B_{\max}$, and $T$.  We use the data structure of
Lemma~\ref{lem:rrr} that for $B_{\min}$ and $B_{\max}$ will occupy

$$O\left(\lg \binom{n}{\frac{n}{ \lg n}} + \frac{n \lg \lg n}{\lg n}\right) = O\left(\frac{n \lg \lg n}{\lg n}\right)$$ 

\noindent
bits, and for $T$ will occupy no more than $2n + \Oh(\frac{n\lg\lg n
}{\lg n})$ bits.  Thus, our data structures at this point occupy
$3n+o(n)$ bits in total, counting the space for $U$.  We will describe
how to recover $\lg n$ consecutive bits of $T_{\min}$; the procedure
for $T_{\max}$ is analogous.  Consider the distances between two
consecutive $1$ bits having indices $x_i$ and $x_{i+1}$ in $B_{\min}$.
Suppose $x_{i+1} - x_{i} \le c \lg n$ in $B_{\min}$, for some constant
$c \ge 9$.  In this case we call the corresponding block $\beta_i$ of
$\lg n$ consecutive bits of $B_{\min}$ \emph{min-good}, and otherwise
we call $\beta_i$ \emph{min-bad}.  We also define similar notions for
\emph{max-good} and \emph{max-bad} blocks.  The problem now becomes
recovering any block (good or bad), since if the $\lg n$ consecutive
bits we wish to extract are not on block boundaries we can simply
extract two consecutive blocks which overlap the desired range, then
recover the bits in the range using bit shifting and bitwise
arithmetic.

If $\beta_i$ is min-good, then we can recover it in $\Oh(c) = \Oh(1)$
time, since all we need to do is scan the corresponding segment of $T$
between the two $1$s, as well as the segment of $U$ starting at
$\texttt{rank}_1(T,x_i)$.  We process the bits of $T$ and $U$ together
in blocks of $\frac{\lg n}{4}$ each, using the lookup table
$\mathcal{L}$: note that we can advance in $U$ correctly by
determining $t$ by counting the number of $1$ bits in either in $s_1$
or $s_3$.  This can be done using either an additional lookup table of
size $\Theta(\sqrt{n})$ using constant time, or by storing the answer
explicitly in $\mathcal{L}$.  When we do this, there is one border
case which we must handle, which occurs when the last bit in $s_1$ is
not a $1$.  However, we can simply append a $1$ to end of $s_1$ in
this case, and then delete either $1$ or $01$ from the end of $s_3$,
depending on the value of $s_2[t]$.  This correction can be done in
$\Oh(1)$ time using bit shifting and bitwise arithmetic.

If $\beta_i$ is min-bad, then we store the answer explicitly.  This
can be done by storing the answer for each bad $\beta_i$ in an array
of size $z \lg n$ bits, where $z$ is the number of bad blocks.  Since
$z \le \lceil \frac{n}{c \lg n} \rceil $ this is $\lceil \frac{n}{c}
\rceil $ bits in total.  We also must store yet another bit vector,
encoded using Lemma~\ref{lem:rrr}, marking the start of the min-bad
blocks, which occupies another $\Oh(\frac{n \lg \lg n}{\lg n})$ bits
by a similar calculation as before.  Thus, we can recover any block in
$B_{\min}$ using $3n + \lceil \frac{n}{c}\rceil + o(n)$ bits in
$\Oh(c) = \Oh(1)$ time.  

In fact, by examining the structure of Lemma~\ref{lem:rrr} in more
detail we can argue that it compresses $T$ slightly for each bad
block, to get a better space bound than $2n+o(n)$ bits. Consider all
the min-bad blocks $\beta_1, ..., \beta_z$ in $B_{\min}$ and the
max-bad blocks $\beta'_1, ..., \beta'_{z'}$ in $B_{\max}$.  For a
given min-bad block $\beta_i$, any max-bad block $\beta'_j$ can only
overlap its first or last $2 \lg n$ bits in $T$.  This follows since
each bit in $T$ corresponds to at least one bit in either $T_{\min}$
or $T_{\max}$, and because less than half of these $2\lg n$ bits can
correspond to bits in $T_{\min}$ (since the block is min-bad). Thus,
each bad block has a middle part of at least $(c-4)\lg n$ bits, which
are not overlapped by any other bad block.  We furthermore observe
that these $(c-4)\lg n$ middle bits are highly compressible, since
they contain at most $\lg n$ one bits, by the definition of a bad
block.  Since these $(c-4)\lg n$ middle bits are compressed to their
zeroth-order entropy in chunks of $\frac{\lg n}{2}$ consecutive bits
by Lemma~\ref{lem:rrr}, we get that the space occupied by each of them
is at most
$$\left\lceil \lg \binom{(c-4)\lg n}{\lg n} \right\rceil + \Theta(c)
\le (c-4)H\left(\frac{1}{c-4}\right) \lg n + \Theta(c) \enspace.$$ 
\noindent
The cost of explicitly storing the answer for the bad block was $\lg
n$ bits. Since $c \ge 9$, and assuming $n$ is sufficiently large, we
get that this additional $\lg n$ bits of space can be added to the
cost of storing the middle part of the bad block in compressed form,
without exceeding the cost of storing the middle part of the bad block
in uncompressed form.  The value of $c \ge 9$ came from a numeric
calculation by finding the first value of $c$ such that
$(c-4)H(\frac{1}{c-4}) + 1 < (c-4)$.  Thus, the total space bound is
$3n+o(n)$ bits.  \qed
\end{proof}

\subsection{Lower Bound for Range Min-Max Queries}

Given a permutation $\pi = (p_1, ..., p_n)$, we say $\pi$ contains the
permutation pattern $s_1\text{-}s_2\text{-}...\text{-}s_m$ if there
exists a subsequence of $\pi$ whose elements have the same relative
ordering as the elements in the pattern.  That is, there exist some
$x_1< x_2 <...< x_m \in [1,n]$ such that for all $i,j \in [1,m]$ we
have that $\pi(x_i) < \pi(x_j)$ if and only if $s_i < s_j$.  For
example, if $\pi = (1,4,2,5,3)$ then $\pi$ contains the permutation
pattern $1\text{-}3\text{-}4\text{-}2$: we use this hyphen notation to
emphasize that the indices need not be consecutive. In this case, the
series of indices in $\pi$ matching the pattern are $x_1 = 1$, $x_2 =
2$, $x_3 = 4$ and $x_4 = 5$.  If no hyphen is present between elements
$s_i$ and $s_{i+1}$ in the permutation pattern, then the indices $x_i$
and $x_{i+1}$ must be consecutive: i.e., $x_{i+1} = x_i +1$.  In terms
of the example, $\pi$ does not contain the permutation pattern
$1\text{-}34\text{-}2$.

A permutation $\pi = (p_1, ... , p_n)$ is a \emph{Baxter permutation}
if there exist no indices $1 \le i < j < k \le n$ such that $\pi(j+1)
< \pi(i) < \pi(k) < \pi(j)$ or $\pi(j) < \pi(k) < \pi(i) < \pi(j+1)$.
Thus, Baxter permutations are those that do not contain $\forbMMTwo$
and $\forbMMOne$.  Permutations with less than $4$ elements are
trivially Baxter permutations, and for permutations on $4$ elements
the non-Baxter permutations are exactly $(2,4,1,3)$ and $(3,1,4,2)$.
Baxter permutations are well studied, and their asymptotic behaviour
is known (see, e.g., OEIS A001181~\cite{OEIS}).

We have the following lemma:

\begin{lemma}\label{lem:baxter}
Suppose $\pi$ is a Baxter permutation, stored in an array $A[1..n]$
such that $A[i] = \pi(i)$.  If an encoding that can recover all range
minimum and maximum queries is constructed on $A$, then $\pi$ can be
recovered from the encoding.
\end{lemma}

\begin{proof}
In order to recover the permutation, it suffices to show that we can
perform pairwise comparisons on any two elements in $A$ using range
minimum and range maximum queries. The proof follows by induction on
$n$.

For the base case, for $n=1$ there is exactly one permutation, so there
is nothing to recover. Thus, let us assume that the lemma holds for all
permutations on less than $n\geq 2$ elements.  For a permutation on $n$
elements, consider the sub-permutation induced by the array prefix
$A[1..(n-1)]$ and suffix $A[2..n]$.  These subpermutations must be
Baxter permutations, since deleting elements from the prefix or suffix
of a Baxter permutation cannot create a $\forbMMTwo$ or a $\forbMMOne$.  Thus,
it suffices to show that we can compare $A[1]$ and $A[n]$, as all
the remaining pairwise comparisons can be performed by the induction
hypothesis.

Let $x = \rmin(A[1..n])$ and $y = \rmax(A[1..n])$ be the indices of
the minimum and maximum elements in the array, respectively. If
$x\in\{1,n\}$ or $y\in\{1,n\}$ we can compare $A[1]$ and $A[n]$, so
assume $x,y\in [2,n-1]$.  Without loss of generality we consider the
case where $x < y$: the opposite case is symmetric (i.e., replacing
$\forbMMOne$ with $\forbMMTwo$), and $x\neq y$ because $n\geq
2$. Consider an arbitrary index $i \in [x,...,y]$, and the result of
comparing $A[1]$ to $A[i]$ and $A[i]$ to $A[n]$ (that can be done by
the induction hypothesis, as $i\in[2,n-1]$).  The result is a partial
order on three elements, and is either:
\begin{enumerate}
  \item One of the two chains $A[1] < A[i] < A[n]$ or $A[n] < A[i] <
    A[1]$, in which case we are done since $A[1]$ and $A[n]$ can be
    compared; or
 \item A partial order in which $A[i]$ is the minimum or maximum
   element, and $A[1]$ is incomparable with $A[n]$.
\end{enumerate}
If we are in the latter case for all $i \in [x,y]$, then let $f(i) =
0$ if $A[i]$ is the minimum element in this partial order, and $f(i) =
1$ otherwise. Because of how $x$ and $y$ were chosen, $f(x) = 0$ and
$f(y) = 1$. If we consider the values of $f(i)$ for all $i \in [x,y]$,
there must exist two indices $i,i+1 \in [x,y]$ such that $f(i)=0$ and
$f(i+1)=1$.  Therefore, the indices $1, i, i+1, n$ form the forbidden
pattern $\forbMMOne$, unless $A[1] < A[n]$.  \qed
\end{proof}

\begin{theorem}\label{thm:min-max-lb}
Any data structure encoding range minimum and maximum queries
simultaneously must occupy $3n - \Theta(\log n)$ bits, for
sufficiently large values of $n$.
\end{theorem}

\begin{proof}
Let $L(n)$ be the number of Baxter permutations on $n$ elements.  It
is known (cf.~\cite{OEIS}) that $\lim_{n \to \infty} \frac{L(n)
  \pi\sqrt{3}n^4}{2^{3n+5}} = 1$. Since we can encode and recover each
one by the procedure discussed in Lemma~\ref{lem:baxter}, our encoding
data structure must occupy at least $\lg L(n) = 3n-\Theta(\log n)$
bits, if $n$ is sufficiently large. \qed
\end{proof}

\section{Optimal Encodings for Top-\texorpdfstring{$k$}{k} Queries\label{sec:top-k}}

In this section we use $\rtopk(A[i..j])$ to denote a range
top-$k$ query on the subarray $A[i..j]$.  The solution to such a query
is an ordered list of indices $\{\ell_1, ..., \ell_k\}$ such that
$A[\ell_m]$ is the $m$-th largest element in $A[i..j]$.

\subsection{Upper Bound for Encoding Top-\texorpdfstring{$k$}{k} Queries\label{sec:encoding-topk}}

Like the encoding for range min-max queries, our encoding for range
top-$k$ queries is based on representing the changes to a certain
structure as we scan through the array $A$.  Each prefix in the array
will correspond to a different structure. We denote the structure,
that we will soon describe, for prefix $A[1..j]$ as $S_k(j)$, for all
$1 \le j \le n$. The structure $S_k(j)$ will allow us to answer
$\rtopk(A[i..j])$ for any $i \in [1,j]$. Our encoding will store the
differences between $S_k(j)$ and $S_k(j+1)$ for all $j \in [1,n-1]$.
Let us begin by defining a single instance for an arbitrary $j$.

We first define the directed graph $G_{j} = (V,E)$ with vertices
labelled $\{1, ..., j\}$, and where an edge $(i',j') \in E$ iff both
$i' < j'$ and $A[i'] < A[j']$ for all $1 \le i' < j' \le j$.  We call
$G_{j}$ the \emph{dominance graph} of $A[1..j]$, and say $j'$
\emph{dominates} $i'$, or $i'$ is dominated by $j'$, if $(i',j') \in
E$.  Next consider the out-degree $d_{j}(\ell)$ of the vertex labelled
$\ell \in [1,j]$ in $G_j$. We define an array $S[1..j]$, where
$S[\ell] = d_{j}(\ell)$ for $1 \le \ell \le j$.  The structure
$S_k(j)$ is defined as follows: take the array $S[1..j]$, and for each
entry $\ell \in [1,j]$ such that $S[\ell] > k$, replace $S[\ell]$ with
$k$.  We use the notation $S_k(j,\ell)$ to refer to the $\ell$-th
array entry in the structure $S_k(j)$. We refer to an index $\ell$ to
be \emph{active} iff $S_k(j,\ell) < k$, and as \emph{inactive}
otherwise.  We note that $S_k(n)$ is reminiscent of the one-sided
top-$k$ structure of Grossi et al.~\cite{GINRS13}.

\begin{figure}
\centering
\includegraphics[width=0.9\textwidth]{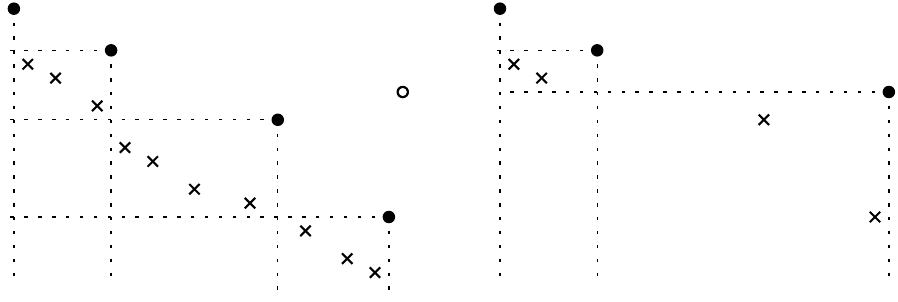}
\caption{\label{fig:update}Geometric interpretation of how the
  structure $S_k(j)$ is updated to $S_k(j+1)$.  In the example $k =
  2$, and the value of each active element in the array is represented
  by its height.  Black circles denote $0$ values in the array
  $S_2(j)$, whereas crosses represent $1$ values, and $2$ values
  (inactive elements) are not depicted.  When the new point (empty
  circle) is inserted to the structure on the left, it increments the
  counters of the smallest $10$ active elements, resulting in the
  picture on the right representing $S_2(j+1)$.}
\end{figure}

\begin{lemma}
\label{lem:active}
The total ordering of elements $A[i_1], ..., A[i_{j'}]$,
where $\{i_1, ..., i_{j'}\}$ are the active indices in
$S_k(j)$, can be recovered by examining only $S_k(j)$.
\end{lemma}

\begin{proof}
We scan the structure $S_k(j)$ from index $j$ down to $1$, maintaining
a total ordering on the active elements seen so far.  Initially, we
have an empty total ordering.  At each active location $\ell$ the
value $S_k(j,\ell)$ indicates how many active elements in locations
$[\ell+1,j]$ are larger than $A[\ell]$.  This follows since an
inactive element cannot dominate an active element in the graph
$G_j$. Thus, we can insert $A[\ell]$ into the current total ordering
of active elements. \qed
\end{proof}

We define the \emph{size} of $S_k(j)$ as follows: $|S_k(j)| =
\sum_{\ell = 1}^{j} (k - S_k(j,\ell))$.  The key observation is that
the structure $S_k(j+1)$ can be constructed from $S_k(j)$ using the
following procedure:

\begin{enumerate}
\item Compute the value $\delta_{j} = |S_k(j)| - |S_k(j+1)| + k$. This
  quantity is always nonnegative, as we add one new element to the
  large staircase, which increases the size by at most $k$.

\item Find the $\delta_j$ indices among the active elements in
  $S_k(j)$ such that their values in $A$ are the smallest via
  Lemma~\ref{lem:active}.  Denote this set of indices as
  $\mathcal{I}$.

\item For each $\ell \in [1,j]$, set $S_k(j+1,\ell) = S_k(j,\ell)+1$
  iff $\ell \in \mathcal{I}$, and $S_k(j+1,\ell) = S_k(j,\ell)$
  otherwise.

\item Add the new element at the end of the array, setting
  $S_k(j+1,j+1) = 0$.
\end{enumerate}

Thus, to construct $S_k(j+1)$ all that is needed is $S_k(j)$ and the
value $\delta_j$: see Figure~\ref{fig:update}.  This implies that by
storing $\delta_j$ for $j \in [1,n-1]$ we can build any $S_k(j)$.

\begin{theorem}
\label{thm:top-k-ub}
Solutions to all queries $\rtopk(A[i..j])$ can be encoded in at most
$(k+1)nH(\frac{1}{k+1})$ bits of space.
\end{theorem}

\begin{proof}
Suppose we store the bitvector $0^{\delta_1}10^{\delta_2}1\ldots
0^{\delta_{n-1}}1$.  This bitvector contains no more than $kn$ zero
bits.  This follows since each active counter can be incremented $k$
times before it becomes inactive.  Thus, storing the bitvector
requires no more than $\lg \binom{(k+1)n}{n} \le
(k+1)nH(\frac{1}{k+1})$ bits.

Next we prove that this is all we need to answer a query $\rtopk(A[i..j])$.  We
use the encoding to construct $S_k(j)$. We know that
for every element at inactive index $\ell$ in $S_k(j)$ there are at
least $k$ elements with larger value in $A[\ell + 1..j]$. Consequently,
these elements need not be returned in the solution, and
it is enough to recover the indices of the top-$k$ values among the elements
at active indices at least $i$. We apply  Lemma~\ref{lem:active} on $S_k(j)$ to
recover these indices and return them as the solution.
\qed
\end{proof}

\subsection{Lower Bound for Encoding Top-\texorpdfstring{$k$}{k} Queries}

The goal of this section is to show that the encoding from
Section~\ref{sec:encoding-topk} is, in fact, optimal. The first
observation is that all structures $S_k(j)$ for $j \in [1,n]$ can be
reconstructed with $\rtopk$ queries.

\begin{lemma}
\label{lem:reconstruction-topk}
Any $S_k(j)$ can be reconstructed with $\rtopk$ queries.
\end{lemma}

\begin{proof}
To reconstruct $S_k(j)$, we execute the query $\rtopk(A[\ell..j])$ for
each $\ell \in [1,j]$.  If index $\ell$ is returned as the $k'$-th
largest element in $[\ell,j]$, then by definition there are exactly
$k' -1$ elements in locations $A[\ell+1 ..j]$ with value larger than
$A[\ell]$.  Thus, $\ell$ is an active location and $S_k(j,\ell) =
k'-1$.  If $\ell$ is not returned by the query, then it is inactive
and we set $S_k(j,\ell) = k$.  \qed
\end{proof}

Recall that we encode all structures by specifying
$\delta_{1},\delta_{2},\ldots,\delta_{n-1}$.  We call an $(n-1)$-tuple
of nonnegative integers $(\delta_{1},\delta_{2},\ldots,\delta_{n-1})$
\emph{valid} if it encodes some $S_k(1),S_k(2),\ldots,S_k(n)$, i.e.,
if there exists at least one array $A[1..n]$ consisting of distinct
integers such that the structure constructed for $A[1..j]$ is exactly
the encoded $S_k(j)$, for every $j=1,2,\ldots,n$. Then the number of
bits required by the encoding is at least the logarithm of the number
of valid $(n-1)$-tuples
$(\delta_{1},\delta_{2},\ldots,\delta_{n-1})$. Our encoding from
Section~\ref{sec:encoding-topk} shows this number is at most
$\binom{(k+1)n}{n}$, but we need to argue in the other direction,
which is far more involved.

Recall that the size of a particular $S_k(j)$ is $|S_k(j)| = \sum_{i =
  1}^{j} (k - S_k(j,i))$.  We would like to argue that there are many
valid $(n-1)$-tuples $(\delta_1,\delta_2,\ldots,\delta_{n-1})$.  This
will be proven in a series of transformations.

\begin{lemma}
\label{lem:valid-tuple-topk}
If $(\delta_{1},\delta_{2},\ldots,\delta_{n-1})$ is valid, then for
any $\delta_{n}\in\{0,1,\ldots,\left\lceil\frac{M}{k}\right\rceil\}$
where $M = \sum_{i=1}^{n-1}(k-\delta_{i})$, the tuple
$(\delta_{1},\delta_{2},\ldots,\delta_{n-1},\delta_{n})$ is also
valid.
\end{lemma}

\begin{proof}
Let $A[1..n]$ be an array such that the structure constructed for
$A[1..j]$ is exactly $S_k(j)$, for every $j=1,2,\ldots,n$. By
definition of $\delta_{j}$, we have that $M =
\sum_{i=1}^{n-1}(k-\delta_{i})< |S_k(n)|$.  Denote the number of
active elements in $S_{k}(j)$ with the corresponding entry set to
$\alpha$ as $m_\alpha$ for $\alpha \in [0,k-1]$. For any $s \in
\{0,1,\ldots, \sum_{\alpha = 0}^{k-1}m_\alpha \}$, we can adjust
$A[n+1]$ so that it is larger than exactly the $s$ smallest active
elements in $S_k(n)$. Thus, choosing any $\delta_{n} \in
\{0,1,\ldots,\sum_{\alpha = 1}^{k}m_\alpha\}$ results in a valid
$(\delta_{1},\delta_{2},\ldots,\delta_{n})$.  Since $|S_k(n)| =
\sum_{\alpha = 0}^{k-1} (k - \alpha)m_\alpha \le k \sum_{\alpha =
  0}^{k-1} m_{\alpha}$, we have $\sum_{\alpha = 0}^{k-1} m_\alpha \geq
\left\lceil\frac{|S_k(n)|}{k}\right\rceil$, proving the claim.  \qed
\end{proof}

Every valid $(n-1)$-tuple $(a_{1},a_{2},\ldots,a_{n-1})$ corresponds
in a natural way to a walk of length $n-1$ in a plane, where we start
at $(0,0)$ and perform steps of the form $(1,a_{i})$, for
$i=1,2,\ldots,n-1$.  We consider a subset of all such walks. Denoting
the current position by $(x_{i},y_{i})$, we require that $a_{i}$ is an
integer from $[k-\left\lceil\frac{y_{i}}{k}\right\rceil,k]$. Under
such conditions, any walk corresponds to a valid $(n-1)$-tuple
$(\delta_{1},\delta_{2},\ldots,\delta_{n-1})$, because we can choose
$\delta_{i} = k - a_{i}$ and apply
Lemma~\ref{lem:valid-tuple-topk}. Therefore, we can focus on counting
such walks.

The condition $[k-\left\lceil\frac{y_{i}}{k}\right\rceil,k]$ is not
easy to work with, though. We will count more restricted walks
instead. A $Y$-restricted nonnegative walk of length $n$ starts at
$(0,0)$ and consists of $n$ steps of the form $(1,a_{i})$, where
$a_{i}\in Y$ for $i=1,2,\ldots,n$, such that the current
$y$-coordinate is always nonnegative. $Y$ is an arbitrary set of
integers.

\begin{lemma}
\label{lem:restricted-walks-topk}
The number of valid $(n-1)$-tuples is at least as large as the number
of $[k -\Delta,k]$-restricted nonnegative walks of length
$n-1-\Delta$.
\end{lemma}

\begin{proof}
We have already observed that the number of valid $(n-1)$-tuples is at
least as large as the number of walks consisting of $n-1$ steps of the
form $(1,a_{i})$, where $a_{i}\in
[k-\left\lceil\frac{y_{i}}{k}\right\rceil,k]$ for
$i=1,2,\ldots,n-1$. We distinguish a subset of such walks, where the
first $\Delta$ steps are of the form $(1,k)$, and then we always stay
above (or on) the line $y=k\Delta$. Under such restrictions, $a_{i}\in
[k-\Delta,k]$ implies $a_{i}\in
[k-\left\lceil\frac{y_{i}}{k}\right\rceil,k]$, so counting
$[k-\Delta,k]$-restricted nonnegative walks gives us a lower bound on
the number of valid $(n-1)$-tuples.  \qed
\end{proof}

We move to counting $Y$-restricted nonnegative walks of length
$n$. Again, counting them directly is non-trivial, so we introduce a
notion of $Y$-restricted returning walk of length $n$, where we ignore
the condition that the current $y$-coordinate should be always
nonnegative, but require the walk ends at $(n,0)$.

\begin{figure}
\centering
\includegraphics[width=\textwidth]{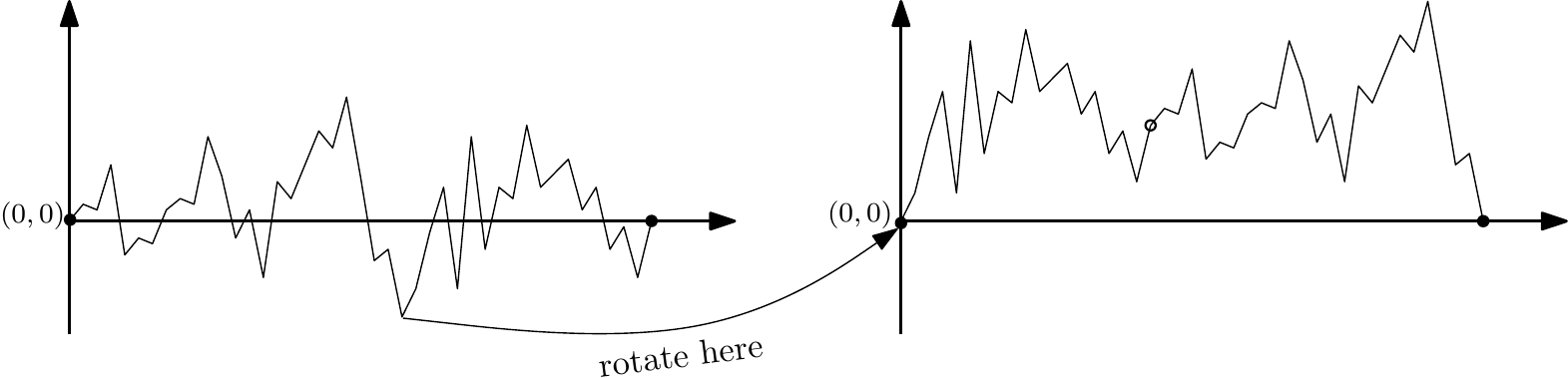}
\caption{\label{fig:cyclic-rotation-topk}Left: a $Y$-restricted walk ending
  at $(n,0)$. Right: a cyclic rotation of the walk on the left such
  that the walk is always nonnegative.}
\end{figure}

\begin{lemma}
\label{lem:cycle-lemma-topk}
The number of $Y$-restricted nonnegative walks of length $n$ is at
least as large as the number of $Y$-restricted returning walks of
length $n$ divided by $n$.
\end{lemma}

\begin{proof}
This follows from the so-called cycle lemma~\cite{DM47}, but we prefer
to provide a simple direct proof.  We consider only $Y$-restricted
nonnegative walks of length $n$ ending at $(n,0)$, and denote their
set by $W_1$. The set of $Y$-restricted returning walks of length $n$
is denoted by $W_2$. The crucial observation is that a cyclic rotation
of any walk in $W_2$ is also a walk in $W_2$. Moreover, there is
always at least one such cyclic rotation which results in the walk
becoming nonnegative (see
Figure~\ref{fig:cyclic-rotation-topk}). Therefore, we can define a
total function $f: W_2\rightarrow W_1$, that takes a walk $w$ and
rotates it cyclically as to make it nonnegative. Because there are
just $n$ cyclic rotations of a walk of length $n$, any element of
$W_1$ is the image of at most $n$ elements of $W_2$ through
$f$. Therefore, $|W_1| \geq \frac{|W_2|}{n}$ as claimed.  \qed
\end{proof}

The only remaining step is to count $[k-\Delta,k]$-restricted
returning walks of length $n-1-\Delta$.  This is equivalent to
counting ordered partitions of $k(n-1-\Delta)$ into parts
$a_{1},a_{2},\ldots,a_{n-1-\Delta}$, where $a_{i}\in [0,\Delta]$ for
every $i=1,2,\ldots,n-1-\Delta$.  This follows since a partition of
size $\ell$ corresponds to a step of size $k - \ell$.

\begin{lemma}
The number of ordered partitions of $N$ into $g$ parts, where every
part is from $[0,B]$, is at least $\binom{N-2g'+g-1}{g-g'-1}$, where
$g'=\left\lfloor\frac{N}{B}\right\rfloor$.
\end{lemma}

\begin{proof}
The number of ordered partitions of $N$ into $g$ parts, where there
are no restrictions on the sizes of the parts, is simply
$\binom{N+g-1}{g-1}$. To take the restrictions into the account, we
first split $N$ into blocks of length $B$ (except for the last block,
which might be shorter). This creates $g'+1$ blocks. Then, we
additionally split the blocks into smaller parts, which ensures that
all parts are from $[0,B]$. We restrict the smaller parts, so that the
first and the last smaller part in every block is strictly positive.
This ensures that given the resulting partition into parts, we can uniquely
reconstruct the blocks. Therefore, we only need to count the number of
ways we can split the blocks into such smaller parts, and by standard
reasoning this is at least $\binom{N-2g'+g-1}{g-g'-1}$.  This follows
by conceptually merging the last element in block $i$ with the first
element in block $i+1$, so that no further partitioning can happen
between them, and then partitioning the remaining set into $g-g'$
pieces. Every such partition corresponds to a distinct restricted
partition obtained by splitting between the merged elements, which creates
$g'$ additional blocks. 
\qed
\end{proof}

We are ready to combine all the ingredients. Setting
$N=k(n-1-\Delta)$, $g=n-1-\Delta$,
$g'=\left\lfloor\frac{k(n-1-\Delta)}{\Delta}\right\rfloor=\left\lfloor\frac{k(n-1)}{\Delta}\right\rfloor-k$
and substituting, the number of bits required by the encoding is:
$$ \lg \binom{N-2g'+g-1}{g-g'-1} > \lg \binom{(k+1)(n-2-\Delta-g')}{n-2-\Delta-g'} \enspace .$$
\noindent
Using the entropy function as a lower bound, this is at least
$(k+1)n'H(\frac{1}{k+1}) - \Theta(\log n')$, where $n' = n-2-\Delta-g'
\geq n(1-\frac{k}{\Delta})+\frac{k}{\Delta}+k - 2-\Delta$.  Thus, we
have the following theorem:

\begin{theorem}
\label{thm:top-k-lb}
For sufficiently large values of $n$, any data structure that encodes
range top-$k$ queries must occupy $(k+1)n'H(\frac{1}{k+1}) -
\Theta(\log n')$ bits of space, where $n' \ge
n(1-\frac{k}{\Delta})+\frac{k}{\Delta}+k - 2-\Delta$, and $\Delta \ge
1$ can be selected to be any positive integer.  If $k = o(n)$, then
$\Delta$ can be chosen such that $\Delta = \omega(k)$ and $\Delta =
o(n)$, yielding that the lower bound is $(k+1)n
H(\frac{1}{k+1})(1-o(1))$ bits.
\end{theorem}

\section{\label{sec:datastructure}Data Structure for Top-\texorpdfstring{$k$}{k} Queries}

In this section we show how to use the encoding of
Section~\ref{sec:encoding-topk} to construct a data structure that
supports top-$k$ queries efficiently.

The high-level idea is to decompose the array into blocks, and
construct a new array by storing the $k$ largest elements in each
block.  Then, we build a naive structure over the new (short) array,
called the macro structure, and additionally store a small separate
structure for every block, called the micro structure. This is a
standard approach in succinct data structures, but as soon as we try
to apply it in the top-$k$ setting, quite a few difficulties
appear. The micro structures should be based on the encoding from
Section~\ref{sec:encoding-topk}, which in turn is based on encoding
how the $S_k(j)$'s change. But these changes can be, in some cases,
very non-local, and hence it is not obvious how the blocks should be
defined. This problem also occurs in, for example, encodings for
balanced parenthesis, where the so-called pioneers approach is
used~\cite{GRRR06}.  Here the situation is even more complex, and we
start with developing an appropriate decomposition through a series of
technical lemmas.  Then, using the decomposition, we construct the
macro structure, which allows us to answer any query spanning more
than one block, and the micro structure, which allows us to answer any
query fully inside a single block.

\subsection{Good Decompositions}

Consider the array $A$, and the structure $S_k(j)$ at each array index
$j \in [1,n]$.  Recall that the structure $S_k(j)$ is an array, where
each entry is an integer drawn from the range $[1,k]$.  For technical
reasons we define $S_k(0)$ to be an empty array.  See
Table~\ref{tab:twoseq} for an example of these definitions for $k =
2$.

\begin{table}
\centering
\caption{\label{tab:twoseq}Suppose $A = \{ 46, 31, 93, 16, 45, 77, 25,
  57, 26\}$.  We give the structures $S_k(j)$ for $A$ in the following
  table. The encoding for $A$ is: $1100110010001100101$.}
\begin{tabular}{|r||c|c|c|c|c|c|c|c|c|}
\hline 
   $i$ & 1 & 2 & 3 & 4 & 5 & 6 & 7 & 8 & 9 \\
\hline
   $A[i]$ & 46 & 31 & 93 & 16 & 45 & 77 & 25 & 57 & 26 \\
\hline
\hline
  $S_2(0,i)$ &  &   &   &   &   &   &   &   &    \\
\hline 
  $S_2(1,i)$ & {\bf 0} &   &   &   &   &   &   &   &    \\
\hline 
  $S_2(2,i)$ & {\bf 0} & {\bf 0} &   &   &   &   &   &   &    \\
\hline 
  $S_2(3,i)$ & {\bf 1} & {\bf 1} & {\bf 0} &   &   &   &   &   &    \\
\hline 
  $S_2(4,i)$ & {\bf 1} & {\bf 1} & {\bf 0} & {\bf 0} &   &   &   &   &    \\
\hline 
  $S_2(5,i)$ & {\bf 1} & {\bf 2} & {\bf 0} & {\bf 1} & {\bf 0} &   &   &   &    \\
\hline 
  $S_2(6,i)$ & {\bf 2} & {\bf 2} & {\bf 0} & {\bf 2} & {\bf 1} & {\bf 0} &   &   &    \\
\hline 
  $S_2(7,i)$ & {\bf 2} & {\bf 2} & {\bf 0} & {\bf 2} & {\bf 1} & {\bf 0} & {\bf 0} &   &    \\
\hline 
  $S_2(8,i)$ & {\bf 2} & {\bf 2} & {\bf 0} & {\bf 2} & {\bf 2} & {\bf 0} & {\bf 1} & {\bf 0} &    \\
\hline 
  $S_2(9,i)$ & {\bf 2} & {\bf 2} & {\bf 0} & {\bf 2} & {\bf 2} & {\bf 0} & {\bf 2} & {\bf 0} & {\bf 0}  \\ 
\hline
\end{tabular}
\end{table}

Let $\opset(i) = \{a_1, ..., a_z\}$ be the set of all indices such
that $S_k(i-1,a_\ell) \neq S_k(i,a_\ell)$ for $1 \le \ell \le z$; this
set will include the index $i$.  Furthermore, define $\opset(i_1,i_2)
= \cup_{i = i_1}^{i_2} \opset(i)$.  In the example, $\opset(5) =
\{2,4,5\}$, and $\opset(5,6) = \{1,2,4,5,6\}$.  Note that the encoding
described in Section~\ref{sec:encoding-topk} is such that $\delta_i =
|\opset(i) \setminus\{i\}|$ for $i \in [1,n]$.

Conceptually, we divide the range $[1,n]$ into disjoint
\emph{even-blocks} of length $B$: $[1,B],[B+1, 2B], ...$, for some
parameter $B \ge 1$ that we will fix later, and without loss of
generality, assume that $B$ divides $n$.  We use the notation
$\mathcal{B}_i$ to denote the range $[Bi + 1, B(i+1)]$ for $i \in [1,
  \frac{n}{B}]$.

Our goal is to decompose the array into a collection of disjoint
\emph{blocks}.  Each block will have the property that it consists of
a range of at most $B$ contiguous array elements, and will be also
contained within at most one even-block.  We refer to blocks that span
a single array element as \emph{singletons}.

Suppose our decomposition $\mathcal{D}$ consists of $h$ blocks,
$\mathcal{G}_1, ... \mathcal{G}_h$, and that block $\mathcal{G}_i$
consists of the contiguous range $[g(i),g(i+1)-1]$ in $A$, where $1
\le i \le h$, $g(1) = 1$, and $g(h+1) = n+1$.  We call $\mathcal{D}$
\emph{good} if:

\begin{enumerate}[label=\bfseries D\arabic*]

\item \label{en:size-const} Size Constraint: the total number of
  blocks is $h = \Oh(\frac{k^2 n }{ B})$.

\item \label{en:weight-const} Weight Constraint: Consider the changes
  in the structures $S_k(g(i)), S_k(g(i)+1), \ldots$ that occur as we
  scan the indices of an arbitrary block $\mathcal{G}_i$, from left to
  right.  A good decomposition has that the number of changes (i.e.,
  increment operations) occurring in the structures as a result of the
  elements in a block is relatively small, if the block is not a
  singleton.  Formally, we have that $\sum_{j = g(i)}^{g(i+1)-1} |
  \opset(j) | \le B$ for $1 \le i \le h$ if $\mathcal{G}_i$ is not a
  singleton.  Note that this implies that the bit string
  $0^{\delta_{g(i)}}10^{\delta_{g(i)+1}}1\ldots0^{\delta_{g(i+1)-1}}1$
  has length at most $B$.

\item \label{en:window-const} Window Constraint: Consider the changes
  in the structures that occur as we process each individual block.
  The indices of the structures that change are located in a
  relatively small range, if the block is not a singleton.  Formally,
  suppose that $\mathcal{G}_i \subseteq \mathcal{B}_t$ for some $t \in
  [1,\frac{n}{B}]$.  Then we have that $(\opset(g(i),g(i+1)-1) \setminus
  \mathcal{B}_t) \subseteq \mathcal{B}_{w}$ for some $w \in [1,t-1]$,
  if the block $\mathcal{G}_i$ is not a singleton.  We call
  $\mathcal{B}_w$ the \emph{window} of block $\mathcal{G}_i$.

\end{enumerate}

The remainder of this section proves that we can construct a good
decomposition.

\begin{lemma}\label{lem:decomposition}
There is a good decomposition $\mathcal{D}$ of the array $A$.
\end{lemma}

\begin{proof}
We describe a procedure for computing a decomposition satisfying these
conditions.  For each position $i \in [1,n]$, we define the
\emph{weight} $w_i = |\opset(i)|$.  The weight of a range in $A$ is
equal to the sum of the weights of the positions it spans. Positions
with weight larger than $B$ are called \emph{fat}, and will be
singletons in our decomposition.  Since each $w_i$ corresponds to
$w_i$ zero bits in the encoding plus one, and there are at most $kn$
zero bits, the number of fat elements is at most $\Oh(\frac{kn}{B})$.

Consider the remaining non-fat elements.  We combine these non-fat
elements into $\Oh(\frac{kn}{B})$ blocks such that the weights of the ranges
is at most $B$.  This can be done by iteratively merging pairs of
blocks (initially blocks are just individual non-fat elements), until
the sum of the weight of any two adjacent blocks exceeds $B$.  When
this happens, every other block will have weight at least $\frac{B}{2}$, and
by the argument above there can be at most $\Oh(\frac{kn}{B})$ such blocks.
Furthermore, we subdivide these blocks along the boundaries of
even-blocks, introducing at most $\Oh(\frac{n}{B})$ additional blocks.  

We refer to the above decomposition as the \emph{initial
  decomposition}.  The initial decomposition satisfies
conditions~\ref{en:size-const} (in fact it has $\Oh(\frac{kn}{B})$ blocks
rather than $\Oh(\frac{k^2n}{B})$), and~\ref{en:weight-const}, but not
necessarily~\ref{en:window-const}.  Thus, we must further refine the
blocks in order to ensure to create a good decomposition.  We do this by
splitting them using an iterative procedure that we now describe.

For each block $\mathcal{G}_i \subseteq \mathcal{B}_t$ in the
initial decomposition, we scan it from left to right, calling the
current position $x_0$.  We will split it into a (potentially large)
number of new blocks.  At each step, there are two cases depending on
whether the set $\opset(g(i), x_0) \setminus \mathcal{B}_t$ is
contained within a single even-block.  
\begin{enumerate}

\item If it is, then we extend the current block which begins at
  position $g(i)$ by adding position $x_0$ to it.

\item If not, then we split the current block between positions $x_0 -
  1$ and $x_0$, i.e., set $g(i+1) = x_0$.  Furthermore, when this
  occurs we make position $x_0$ a singleton block. We then recursively
  apply the same procedure to the remaining unscanned part of the
  block adjusting the parameters appropriately.  Thus, we have
  introduced two additional blocks.
\end{enumerate}

Such a refinement clearly has the desired window property.  However,
the difficulty is arguing that the second case only occurs $\Oh(k^2
\frac{n}{B})$ times.  To show this, we use a charging argument in which each
split is charged to the rightmost even-block $\mathcal{B}_{w}$
containing a position in $\opset(x_0) \setminus \mathcal{B}_t$.  We
will bound the number of times $\mathcal{B}_w$ can be charged for a
split by $\Oh(k^2)$.

We say a position is $y$-active if it is active in structure $S_k(y)$.
Consider the $(x_0 - 1)$-active elements immediately before a split
occurs.  Consider the position $a \in \mathcal{B}_w$ such that $a \in
\opset(x_0)$ and $A[a]$ is maximum.  We have that $S_k(x_0-1,a) < k$
since $a$ is, by definition, $(x_0-1)$-active.  Moreover, since a
split occurred, there must be some block $\mathcal{B}_{w'}$ where $w'
< w$ containing a position $a' \in \mathcal{B}_{w'}$ such that $a' \in
\opset(x_0)$.  Since $a'$ is also $(x_0 -1)$-active this implies that
there are at most $k-1$ $(x_0-1)$-active positions contained in
$\mathcal{B}_w$, whose corresponding elements have values larger than
$A[a]$.  Thus, when a split occurs, all but at most $k-1$ of the
$(x_0-1)$-active locations contained in $\mathcal{B}_w$ are
incremented in $S_k(x_0)$.  Furthermore, any location not incremented
must be among the $k-1$ largest values in $A[Bw + 1, B(w+1)]$.  Thus,
after $k$ split operations, all but the $k-1$ largest active locations
become inactive.  Since each split increments at least one location in
$\mathcal{B}_w$ at most $k(k-1)$ additional splits occur before all
elements in $\mathcal{B}_w$ become inactive.  Overall, at most $k +
k(k-1) = \Oh(k^2)$ splits can occur before all elements in
$\mathcal{B}_w$ become inactive.

Since there are $\frac{n}{B}$ even-blocks, we have that the total number of
blocks created by splits (or otherwise) is $\Oh(\frac{k^2 n}{B})$, completing
the proof. \qed
\end{proof}

\subsection{Navigating the Encoding}

Before discussing the data structures we store, we require an
additional result, called an \emph{indexable dictionary},
by Raman, Raman, and Rao~\cite{RRR07}:

\begin{lemma}[\cite{RRR07}]\label{lem:rrr-id} Let $\mathcal{V}$ be a
  bit vector of length $n$ bits, containing $m$ one bits.  In the
  word-RAM model with word size $\Theta(\lg n)$ bits, there is a data
  structure of size $\lg \binom{n}{m} + \Oh(m) + \Oh(\lg \lg n) \le n
  H(\frac{m}{n}) + \Oh(m) + \Oh(\lg \lg n))$ bits that supports the
  following operations in $\Oh(1)$ time, for any $i \in [1,n]$:
\begin{enumerate}
\item $\accessop(\mathcal{V}, i)$: return the bit at index $i$ in $\mathcal{V}$.
\item $\rankop_1(\mathcal{V}, i)$: return the number of bits with
  value $1$ in $\mathcal{V}[1..i]$, iff $\accessop(\mathcal{V},i) =
  1$.  If $\accessop(\mathcal{V},i) = 0$, then a flag is returned
  indicating that the operation cannot be supported.
\item $\selop_1(\mathcal{V}, i)$: return the index of the $i$-th
  bit with value $1$.
\end{enumerate}
\end{lemma}

We apply Theorem~\ref{lem:decomposition} to partition $A$ into
$\Oh(\frac{nk^2}{B})$ blocks $\mathcal{G}_{1},\mathcal{G}_{2},\ldots$,
where $B$ is some parameter that will be fixed later on. We then
construct the following indexes:

\begin{enumerate}

\item \emph{Block index:} This is the rank/select data structure of
  Lemma~\ref{lem:rrr} constructed on a bit vector of length $n$
  marking the block boundaries.  This allows us to find the start of
  an arbitrary block in constant time. This bit vector occupies:
$$\lg\binom{n}{\frac{nk^2}{B}} + \Oh\left(\frac{n \lg \lg n }{ \lg n}\right) \le nH\left(\frac{k^2}{B}\right) +
  \Oh\left(\frac{n \lg \lg n }{ \lg n}\right)$$
\noindent 
bits of space, by Lemma~\ref{lem:rrr}.

\item \emph{Encoding index:} Consider the bit vector storing the
  encoding $E$ (described in Section~\ref{sec:encoding-topk}) on
  $A$. For each zero bit in the encoding $E$, we say that bit is
  \emph{associated} with the one bit immediately to its right.  That
  is, the zero bit at position $i$ is associated with the one bit in
  position $\selop_1(E,\rankop_1(E,i) + 1)$.  Since the $j$-th one bit
  in the encoding is representing element $A[j]$, each zero bit
  associated with this one bit can also be said to be associated with
  $A[j]$. Suppose $A[j]$ is part of a block $\mathcal{G}_i$ which is
  contained in even-block $\mathcal{B}_t$ and has a window contained
  in even-block $\mathcal{B}_w$.  The $0$ bits associated with
  position $A[j]$ come in exactly two flavors:

  \begin{enumerate}
    \item Internal increment: if the $0$ bit corresponds an increment
      operation in $S_k(j)$ that occurs inside even-block ${B}_t$
    \item Window increment: if the $0$ bit corresponds an increment
      operation in $S_k(j)$ that occurs inside the window ${B}_w$
  \end{enumerate}
  
  Suppose that for each $j \in [1,n]$ we create two bit vectors
  $E_{\text{INT}}(j)$ and $E_{\text{WIN}}(j)$.  These two bit vectors
  will be of the form $0^\alpha 1$ and $0^\beta 1$, respectively,
  where $\alpha_j$ is the number of internal increments associated
  with position $j$ and $\beta_j$ is the number of window increments
  associated with position $j$.  Note that $\delta_j = \alpha_j +
  \beta_j$. Let $E_{\text{INT}}$ and $E_{\text{WIN}}$ be the
  concatenation of the $E_{\text{INT}}(j)$ and $E_{\text{WIN}}(j)$ bit
  vectors, respectively.  Both of these bit vectors together have $2n$
  one bits, and $kn$ zero bits.  Thus, storing $E_{\text{INT}}$ and
  $E_{\text{WIN}}$ in the smaller of the two representations discussed
  (either Lemma~\ref{lem:rrr} or \ref{lem:rrr-id}) will occupy $(k+2)n
  H(\frac{2}{k+2}) + \Oh(\min\{\frac{nk \lg\lg(nk) }{\lg (nk)},n\})$ bits in
  total.  Note that we cannot perform rank operations on arbitrary
  positions in these bit vectors using the bound just stated, though
  we can perform arbitrary select operations.


\end{enumerate}

\begin{lemma}\label{lem:subarray-enc}
Using the above data structures we can recover the length $j'$ suffix
of the structure $S_k(Bi + j)$ for any $i \in [1,\frac{n}{B}]$, $j \in
[1,B-1]$ and $j' \in [1,j]$ in $\Oh(B^2)$ time.
\end{lemma}

\begin{proof}
Let $m_\alpha = \selop_1(E_{\text{INT}},\alpha)$, and consider the
range of $E_{\text{INT}}$ between $[m_{Bi-1}+1,m_{Bi+j}]$.  This range
contains all the one bits in $E_\text{INT}$ associated with elements
$A[Bi], A[Bi + 1], \ldots, A[Bi+j]$, and also the zero bits associated
with these one bits that are internal increments.  Furthermore the
length of this range in $E_\text{INT}$ is at most $\Oh(B^2)$, since
even in the case where every position in the even-block is a
singleton, the number of internal increments for each of these is
upper bounded by the length of the even-block, $B$.  Thus, to recover
the fragment of the structure $S_k(Bi+j)$, we construct an array of
length $j +1$ in which each index stores a $\lceil \lg (k+1) \rceil$
bit number.  We process the internal increments associated with the
elements $A[Bi]$, $A[Bi +1]$, ..., etc. in order, calling the current
position $\ell$, where $Bi \le \ell \le Bi +j$. We maintain the total
ordering over all currently active elements in length $\ell - Bi +1$
suffix of $S_k(\ell)$ as follows.  Suppose position $\ell$ is
associated with $x$ internal increments (we can determine this by
comparing $m_{\ell}$ and $m_{\ell-1}$ in $\Oh(1)$ time using the
encoding index).  We insert position $\ell$ into the total order as
the $(x +1)$-th smallest element, set its counter value to $0$, and
increment the counters associated with the $x$ smallest elements.  If
an incremented counter exceeds $k -1$, then we remove it from the
total order.  Maintaining the total ordering as a linked list is
sufficient to process the fragment in $\Oh(x)=\Oh(B)$ time per
position $\ell$.  Since there are at most $\Oh(B)$ positions, the
total time is $\Oh(B^2)$. \qed
\end{proof}

\begin{lemma}\label{lem:subarray-query}
Given a subarray $A[x_1..x_2]$ that is contained within a block, we
can return a list $L$ such that $L[p]$ stores the position of the
$p$-th largest element in $A[x_1..x_2]$ in $\Oh(B^2)$ time.
\end{lemma}

\begin{proof}
Using Lemma~\ref{lem:subarray-enc} we can build the length $\ell =
x_2-x_1 +1$ suffix of the structure $S_k[x_2]$, which stores $\ell$
$\lceil \lg (k+1) \rceil$-bit numbers.  Once we have this length
$\ell$ array, we scan it from right to left, constructing the total
order of elements in $A[x_1 ..x_2]$ by Lemma~\ref{lem:active}.  As
before, using a linked list to store the total order is sufficient to
achieve the claimed time bound. \qed
\end{proof}

\subsection{Version Control}

One issue that arises is that to answer queries we will need to
construct fragments of the structure $S_k(j)$ for various values of
$j$, which are not necessarily short suffixes.  In particular, given a
block $\mathcal{G}_i$, we wish to be able to reconstruct its
\emph{window fragment}, which is the fragment of the structure
$S_k(g(i) - 1)$ corresponding to the window of block $\mathcal{G}_i$.
Suppose the window is even-block
$\mathcal{B}_w$. Lemma~\ref{lem:subarray-enc} only allows us to
construct the length $B$ suffix of structure $S_k(B(w+1))$, rather
than the window fragment of $\mathcal{G}_i$.  Thus, we are interested
in how much space is required to recover a window fragment given what
we can recover using Lemma~\ref{lem:subarray-enc}.

\begin{lemma}
Suppose block $\mathcal{G}_i$ has window $\mathcal{B}_w$.  The
difference $\texttt{diff}(i)$ between the window fragment of
$\mathcal{G}_i$ and the length $B$ suffix of $S_k(B(w+1))$ can be
stored using $\Theta(k \lg (B+1))$ bits.  Using $\texttt{diff}(i)$, in
addition to the other data structures described thus far, we can
construct the window fragment of $\mathcal{G}_i$ in time $\Oh(B^2)$.
\end{lemma}

\begin{proof}
Lemmas~\ref{lem:subarray-enc} and~\ref{lem:active} allow us to recover
the total order $\mathcal{L}$ of the $(B(w+1))$-active elements in the
window fragment in $\Oh(B^2)$ time.  Consider the sequence of
positions in the array $A$, $\{ x_1, ..., x_z \} $ that have window
increments associated with them occurring within the window fragment,
where $x_1 > B(w+1)$ and $x_z < g(i)$.  Each element $A[x_\ell]$ can
be mapped to a position $y_\ell$ in the total order $\mathcal{L}$.  It
is sufficient to record the $k$ largest values in this mapping, as all
$B(w+1)$-active positions represented in $\mathcal{L}$ which are
smaller than the $k$-th largest such value will become
$(g(i)-1)$-inactive.  Storing how these $k$ values interleave with the
ordering $\mathcal{L}$ requires at most $k \lceil \lg (B +1) \rceil$
bits of space.  Note that we do not need to know the positions where
these elements occur in $A$ in order to reconstruct the window
fragment, just their positions in the total ordering $\mathcal{L}$,
which contains at most $B$ elements. \qed
\end{proof}

We store $\texttt{diff}(i)$ for each $i \in [1,h]$ (recall $h$ is the
number of blocks).  This requires $\Oh(hk\lg B) =
\Oh(\frac{nk^3\lg(B+1)}{B})$ bits of space in total.

\subsection{Decomposing Queries}

Any range top-$k$ query is either fully within a single block, or
consists of three parts: a suffix of a block $\mathcal{G}_{i}$ that we
call the \emph{left part}, then a number of full blocks
$\mathcal{G}_{i+1},\ldots,\mathcal{G}_{j-1}$ that we call the
\emph{middle part}, and finally a prefix of a block $\mathcal{G}_{j}$
that we call the \emph{right part}.  Note that any of these three
parts may be an empty range. Using the block index we can determine
these parts in $\Oh(1)$ time.

We construct a new array $A'$ by keeping the $k$ largest elements from
every block (if a block is a singleton, this is just one element) and
normalizing all the elements by sorting. $A'$ is stored explicitly and
augmented with a range maximum query structure, which allows us to
locate the $k$ largest element in any query range via a three-sided
range reporting query: this can be done in $\Oh(k)$ time and
$\Oh(\frac{nk^3\lg n}{B})$ bits of space, using successive queries to a range
maximum structure built over $A'$ since we have access to these
elements.

Additionally, for every $j$ such that $A[j]$ appears in $A'$, i.e., is
one of the $k$ largest elements in its block, we store the positions
of the first $k$ larger elements on its left in $A$. This requires
space $\Oh(\frac{nk^3 \lg n }{ B})$ bits.

\subsection{Wrap Up}

Now that we have described all of the data structures, we can explain
how to extract the positions of the top-$k$ elements, given a query
range $A[i..j]$.  The algorithm will consist of first finding the
positions of the top-$k$ elements in the middle part, and the total
ordering of elements in the left and right parts.  Extracting the
solution from the middle part is trivial, since we have a top-$k$ data
structure explicitly stored on the top-$k$ elements in each block.
Extracting the total ordering of elements from the left (or right)
part can be done by applying Lemma~\ref{lem:subarray-query} to the
even-block containing the left or right part.

At this point, we have at most three lists $L_1$,$L_2$, and $L_3$,
storing positions of the top elements from the left, middle, and right
parts respectively, i.e., $L_p[q]$ is the position of the $q$-th
largest element in list $p$. We now argue that we can merge these
lists.

\begin{lemma}
Suppose we are given a query $A[i..j]$.  A list $L$ can be constructed
such that $L[q]$ is the position of the $q$-th largest element, for $1
\le q \le k$, in $A[i..j]$ in time $\Oh(k + B^2)$.
\end{lemma}

\begin{proof}
First, we construct the three lists $L_1$, $L_2$ and $L_3$ as
described above in time $\Oh(B^2)$.  Then we merge the lists $L_1$
(the left part) with the list $L_2$ (the middle part).  This is done
by examining the left pointers of each position in $L_2$.  Consider
the subset $\{\Upsilon_1, ..., \Upsilon_{k'}\}$ of positions in the
left part such that $L_2[p]$ has a left pointer to $\Upsilon_{r}$, for
$1 \le r \le k'$.  If $p + k' \le k$, then implies that $L_2[p]$ is
the $(p + k')$-th largest element in the combination of the left and
middle parts.  Otherwise, it implies that $L_2[p]$ is not in the
top-$k$ in the combination of the two parts.  Using this procedure we
merge the lists $L_1$ and $L_2$, calling the result $L'$.

Next we describe how to merge $L'$ and $L_3$ (the right part).  Recall
that the right part is a prefix of some block $\mathcal{G}_{r}$.  We
reconstruct the window fragment of $\mathcal{G}_{r}$ using
$\texttt{diff}(r)$.  We then scan through $\mathcal{G}_{r}$ up to
position $j$, performing window increments on the window fragment by
reading $E_{\text{WIN}}$.  Let $m_\alpha =
\selop_1(E_{\text{WIN}},\alpha)$.  We read the window increments of
$E_{\text{WIN}}$ from the range $[m_{g(r)-1}+1,m_{j}]$.  Since
$\mathcal{G}_{r}$ is not a singleton block (otherwise it would be
fully contained in the middle part), we have that the length of this
range in $E_\text{WIN}$ is bounded by $B$.  We process the window
increments in order to reconstruct the range $\mathcal{B}_w$ spanned
by the window of $\mathcal{G}_r$ in the structure $S_k(j)$.  During
this process, considering a position $j' \in [g(r),j]$, we observe
that if $L'[p] \not \in \mathcal{C}(j')$, then $A[L'[p]] > A[j']$,
unless position $L'[p]$ had been made inactive by a previous window
increment earlier in the process.  If $L'[p]$ is not in
$\mathcal{B}_w$, then we can infer that $A[L'[p]] > A[j']$
immediately.  Thus, it is possible to insert the positions
$g(r),\ldots,j$ into the list $L'$ to construct the final list $L$
containing the top-$k$ positions in $A[i..j]$.  \qed

\end{proof}

From the above lemmas, we immediately get the following theorem:

\begin{theorem}
There is a data structure occupying
$$(k+2)n H\left(\frac{2}{k+2}\right) + nH\left(\frac{k^2}{B}\right) + \Oh\left(\frac{k^3 n \lg
  n}{ B} + \min\left\{\frac{nk \lg\lg(nk) }{\lg (nk)},n\right\}\right)$$ bits of space,
and supports range top-$k$ queries in $\Oh(k + B^2)$ time.
\end{theorem}

By setting $B = k^{3}\lg n \sqrt{f(n)}$, for a strictly increasing
function $f$, we get the following result:

\begin{corollary}
For any strictly increasing function $f$, there is a data structure
occupying $(k+2)n H(\frac{2}{k+2}) + o(n \lg k)$ bits of space, and
supports range top-$k$ queries in $\Oh(k^6 \lg^2 n f(n))$ time.
\end{corollary}

\subsection{Improvement to Space Bound}

Our final theorem argues that we can slightly improve the space bound:

\begin{theorem}
\label{thm:ds-final} 
For any strictly increasing function $f$, there is a data structure
occupying $(k+1.5)n H\left(\frac{1.5}{k+1.5}\right) + o(n \lg k)$ bits
of space, and supports range top-$k$ queries in $\Oh(k^6 \lg^2 n
f(n))$ time.
\end{theorem}

\begin{proof}
We observe that we need not store a $1$ in the bit vector
$E_{\text{WIN}}$ for elements that are not in the top-$k$ of their
prefix of their even block, as such elements perform no window
increments.  Initially, this does not seem to buy us anything, since
every position can be in the top-$k$ of the prefix of its even block,
but in this case we can take the reversal of the array.  We call an
element \emph{bad} if it is in the top-$k$ of the prefix or suffix of
its even block, and \emph{good} otherwise.

To bound the number of bad positions, consider the top-$2k$ elements
in each even-block.  No other elements can be bad, since there is a
subset of size at least $k$ of these top-$2k$ elements on either its
right or left.  Next consider a good element.  It can only contribute
a one bit to $E_{\text{WIN}}$ in $A$ or to the reverse of $A$, but not
both.  Thus, we have $n - \frac{2kn}{B}$ elements contributing $n -
\frac{2kn}{B}$ one bits to the window encodings for either $A$ or its
reverse.  We therefore need only record $\frac{n}{2} - \frac{2kn}{2B}
+ \frac{2kn}{B} = \frac{n}{2} - o(n)$ one bits for the window encoding
bit vector of $A$ or its reverse.  This reduces the leading term of
the space cost to $(k + 1.5)n H(\frac{1.5}{k+1.5})$.  To correct for
the fact that we have removed one bits from $E_{\text{WIN}}$, we must
adjust select operations on this bit vector by explicitly storing, for
each block, how many elements in the block are good.  Then, when we
process the window increments in a block, we can determine whether an
element is good by examining its internal increments.  This adds an
overhead of $\Oh(\frac{n k^2 \lg n }{ B}) = o(n)$ bits of space and
adds an $\Oh(B^2)$ time cost for determining which elements are good
in a block. \qed
\end{proof}

\bibliographystyle{splncs03}
\bibliography{biblio}

\newpage

\appendix

\section{\label{sec:reduction}Lower bound for Unsorted Range Top-\texorpdfstring{$k$}{k}}

Let a \emph{sorted} range top-$k$ query denote the problem of
returning the indices $i_1, \ldots, i_k$ of the $k$ largest values in
a query range $[i,j]$, in ascending order: i.e., $A[i_j]$ is the
$j$-th largest value.  Let an \emph{unsorted} range top-$k$ query
denote the weaker query in which the indices $i_1, \ldots, i_k$ are
returned in an arbitrary order.

\begin{lemma}
If $\mathcal{S}(n,k)$ is the number of bits required to store an
encoding of sorted range top-$k$ queries on an array $A[1..n]$, then
at least $\mathcal{S}(n - k,k)$ bits are required to store an encoding
of unsorted range top-$k$ queries.
\end{lemma}

\begin{proof}
Suppose there exists an encoding for unsorted range top-$k$ queries
that requires strictly less than $\mathcal{S}(n-k,k)$ bits.  We will
show that such an encoding can be used to construct an encoding for
sorted range top-$k$ queries that occupies strictly less than
$\mathcal{S}(n,k)$ bits.  We pad the input array $A[1..n]$ with $k$
additional values $A[n+1], \ldots, A[n+k]$ such that $A[n+i] > A[j]$
for all $i \in [1,k]$ and $j \in [1,n]$.  We now claim that the
unsorted encoding for the padded array can be used to recover
solutions to all sorted range top-k queries on ranges in $[1,n]$.
Given a query range $[i,j]$, we examine the solutions to unsorted
range top-$k$ queries $[i, j], [i,j+1], \ldots , [i,n+k]$.  Let
$\kappa(j')$ denote the set of indices in $[i,j']$, $\kappa_0 =
\kappa(j)$, $\kappa_\ell = \kappa(\ell')$ where $\ell'$ is the minimum
index such that $\kappa(\ell'-1) \neq \kappa_{\ell_0}$, for $\ell \in
[1,k]$.  By the method we use to pad $A$, it implies that $\kappa_{k}
\cap \kappa_0 = \emptyset$, since the solution to query $[i,n+k]$ is
the set of indices in $[n+1,n+k]$.  Thus, the index of the $k-i$-th
largest element in the sorted solution can be extracted by computing
$\kappa_i \setminus \kappa_{i+1}$ for $i \in [0,k-1]$.  This follows
since the smallest element in $\kappa_i$ is removed, and a new
elemented added to create $\kappa_{i+1}$.  Therefore, we have a
contradiction, since any encoding for the sorted variant must occupy
$\mathcal{S}(n,k)$ bits, and we have given an encoding that occupies
strictly less than $\mathcal{S}(n+k-k,k) = S(n,k)$ bits. \qed
\end{proof}

Thus, for $k=o(n)$ the previous lemma, combined with
Theorem~\ref{thm:top-k-lb} (which provides the function
$\mathcal{S}(n,k)$), implies that the space required for the unsorted
encoding on an array of $n$ elements is within additive lower order
terms of the space required for the sorted encoding on $n$ elements.
\end{document}